\documentclass[11pt,letterpaper]{article}
\usepackage[margin=1in]{geometry}
\usepackage{graphicx}
\usepackage[square,sort]{natbib}
\usepackage{amsmath,amsfonts,amssymb,amsthm,mathtools}
\usepackage{enumitem} %% [label=(\arabic*)]
\usepackage{xspace}
\usepackage[dvipsnames]{xcolor}
\definecolor{ptblue}{RGB}{15,76,129}
\definecolor{ptemerald}{HTML}{009473}
\usepackage{hyperref}
\hypersetup{
linktocpage,
colorlinks=true,
citecolor=ptemerald,
urlcolor=ptblue,
linkcolor=Plum,
}
\usepackage[switch]{lineno}

\usepackage[ruled,linesnumbered,vlined]{algorithm2e} % For algorithms

\usepackage[normalem]{ulem}

\newtheorem{theorem}{Theorem}[section]
\newtheorem{lemma}[theorem]{Lemma}

\newtheorem{corollary}[theorem]{Corollary}

\theoremstyle{definition}
\newtheorem{definition}[theorem]{Definition}

\theoremstyle{remark}

\newcommand{\amms}{$\alpha$-MMS\xspace}

\newcommand{\weightedPROP}{\textsc{WPRAlloc}\xspace}
\newcommand{\mixMMS}{\text{MMS}\xspace}
\newcommand{\homoAlg}{\text{\sc Mixed-MMS-Homogeneous}\xspace}

\title{Maximin Fairness with Mixed Divisible and Indivisible Goods\thanks{A preliminary version appeared in Proceedings of the 35th AAAI Conference on Artificial Intelligence (AAAI)~\citep{BeiLiLu21}. This project/research is supported by the Ministry of Education, Singapore, under its
Academic Research Fund Tier 2 (MOE2019-T2-1-045).}}

\author{ %% Names
Xiaohui Bei\\
Nanyang Technological University\\
\texttt{xhbei@ntu.edu.sg}
\and
Shengxin Liu\\
Harbin Institute of Technology, Shenzhen\\
\texttt{sxliu@hit.edu.cn}
\and
Xinhang Lu\\
Nanyang Technological University\\
\texttt{xinhang001@e.ntu.edu.sg}
\and
Hongao Wang\\
Nanyang Technological University\\
\texttt{hongao.wang@ntu.edu.sg}
}

\date{}
\begin{document}
\maketitle

\begin{abstract}
We study fair resource allocation when the resources contain a mixture of divisible and indivisible goods, focusing on the well-studied fairness notion of maximin share fairness (MMS).
With only indivisible goods, a full MMS allocation may not exist, but a constant multiplicative approximate allocation always does.
We analyze how the MMS approximation guarantee would be affected when the resources to be allocated also contain divisible goods.
In particular, we show that the worst-case MMS approximation guarantee with mixed goods is no worse than that with only indivisible goods.
However, there exist problem instances to which adding some divisible resources would strictly decrease the MMS approximation ratios of the instances.
% One may intuitively think that with some divisible goods, the approximation ratio would be better than that allocating only the indivisible goods; we show that this is not the case.
% Nevertheless, the worst-case approximation guarantee in the mixed goods setting is at least the same as that in the indivisible goods setting.
On the algorithmic front, we propose a constructive algorithm that will always produce an \amms allocation for any number of agents, where $\alpha$ takes values between $1/2$ and $1$ and is a monotonically increasing function determined by how agents value the divisible goods relative to their MMS values.
% To circumvent the computational issues, we then design an approximation algorithm to compute a $(1-\epsilon)\alpha$-MMS allocation in time polynomial in $n, m, L$ and exponential in $1/\epsilon$, where $n$ is the number of agents, $m$ is the number of indivisible goods, $L$ is the number of bits needed to represent an input.
\end{abstract}

\newpage

\section{Introduction}
Fair division concerns the problem of allocating a set of goods among interested agents in a way that is \emph{fair} to all participants involved.
% Two of the most prominent fairness notions in the literatures are \emph{proportionality} and \emph{envy-freeness}, introduced by \citet{Steinhaus49} and \citet{Foley67}, respectively.
The goods involved could be heterogeneous and \emph{divisible}, usually modelled by a \emph{cake}, in which case the problem is also known as \emph{cake-cutting}; in some other cases, the goods are heterogeneous and \emph{indivisible}, and the problem is known as \emph{indivisible resource allocation}.

Due to its subjective nature, a plethora of fairness notions have been proposed and investigated in different resource allocation scenarios (see, e.g.,~\citep{Young95,BramsTa96,BrandtCoEn16,Endriss17,Moulin19} for a survey).
In particular, as one of the most classic and widely known fairness notions, \citet{Steinhaus49} proposed that in an allocation that involves $n$ participating agents, each agent should receive a bundle which is worth at least $1/n$ of her value for the entire set of goods.
An allocation satisfying such property is then known as a \emph{proportional} allocation.
Moreover, \citet{Steinhaus49} also showed that a proportional allocation can always be found for any number of agents over any divisible good.\footnote{Such a solution was credited to B. Knaster and S. Banach by \citet{Steinhaus49}.}
% An allocation is envy-free if each agent prefers her own bundle to any bundle of other agent.
% While both notions can be fulfilled for \emph{divisible} goods \citep{Alon1987},
However, this is not the case when goods are \emph{indivisible}, with the simplest counterexample of two agents dividing a single valuable good.
In order to circumvent this issue, \citet{Budish11} presented a natural alternative to the classic proportionality notion that also works for indivisible goods, known as the \emph{maximin share (MMS) guarantee}.
% Envy-freeness (EF) can be relaxed to \emph{envy-freeness up to one good (EF1)} which requires that no agent prefers the bundle of other agent following the removal of some good from the latter bundle.
% EF1 allocations always exist for any number of agents \citep{LiptonMaMo04}.
In this definition, the \emph{maximin share (MMS)} of an agent is defined as the largest value that she can get if she is allowed to partition goods into $n$ bundles and always receives the least desirable bundle.
An allocation is said to be an \emph{MMS allocation} if every agent receives a bundle which is worth at least her maximin share.

The notion of MMS nicely captures the local measure of fairness even when the goods to be allocated are indivisible.
A natural question then arises of whether an MMS allocation always exists in all problem instances.
Surprisingly, \citet{KurokawaPrWa18} showed that even with additive valuation functions, an MMS allocation may not always exist when there are at least three agents.\footnote{An MMS allocation always exists when there are two agents~\citep{BouveretLe16}.}
However, \citeauthor{KurokawaPrWa18} showed that a $2/3$-MMS allocation can always be found, in which each agent is guaranteed to receive a bundle worth at least $2/3$ of their MMS value.
In other words, if we define the \emph{MMS approximation guarantee} of a problem instance as the largest $\alpha$ such that the instance admits an $\alpha$-MMS allocation, the results in \citep{KurokawaPrWa18} imply that the worst MMS approximation guarantee across all indivisible problem instances is strictly less than 1 and at least 2/3.
Since then, many subsequent works have been carried out on the improvements of MMS approximation guarantee (for which the current best ratio is $3/4 + 1/(12n)$ due to \citet{GargTa20}), design of simpler algorithms, etc.~\citep{AmanatidisMaNi17,BarmanKr20,GhodsiHaSe21,GargMcTa19}.
MMS has also been adopted as the fairness solution concept in several practical applications~\citep{Budish11,GoldmanPr15}.

Even though MMS has been mainly studied in the context of indivisible resource allocation, it is also a well-defined fairness notion in a more general setting where both divisible and indivisible goods are to be allocated.
Many real-world scenarios, including but not limited to divorce or inheritance settlements, involve allocating simultaneously divisible goods such as land or money and indivisible goods such as houses or cars.
What fairness notion should one adopt when dividing resources of such mixed types?
The problem of fairly allocating mixed divisible and indivisible goods was first studied by~\citet{BeiLiLi21}, in which the authors proposed a new fairness notion called \emph{envy-freeness for mixed goods (EFM)} that generalizes envy-freeness, another well-studied fairness notion, to the mixed goods setting.
The maximin share guarantee, on the other hand, can be directly applied to the mixed goods setting without any modification.
This allows us to compare the results of MMS for mixed goods directly to those for indivisible goods.

In this paper, we aim to provide such a comparison.
More specifically, we extend the analysis of MMS allocations to the setting with mixed types of goods, and study its existence, approximation, as well as computation.
In particular, we hope to answer the following questions:
\begin{enumerate}
\item Is the worst-case MMS approximation guarantee across all mixed goods instances the same as that across all indivisible goods instances?

\item Given any problem instance, would adding some divisible resources to it always (weakly) increase the best possible MMS approximation ratio of this instance?

\item How to design algorithms that could find allocations with good MMS approximation guarantee in mixed goods problem instances?
\end{enumerate}

\subsection{Our Results}
In this paper, we answer the three questions posed above.

% Our first set of results is on the relations of the worst-case MMS approximation guarantee in the mixed goods setting and that in the indivisible goods setting.
In Section~\ref{sec:existence}, we first show that any problem instance of mixed goods can be converted into another problem instance with only indivisible goods, such that the two instances have the same MMS value for every agent, and any allocation of the indivisible instance can be converted into an allocation in the mixed instance.
This reduction directly implies that the worst-case MMS approximation guarantee across all mixed goods instances is the same as that across all indivisible goods instances.
% In other words, in terms of worst-case MMS approximation guarantee, having mixed types of goods will not make things worse compared to having only indivisible goods.

This is not a surprising result, because the non-existence of MMS allocations only arises when the resources to be allocated become indivisible.
It is therefore reasonable to think that adding divisible goods to the set of indivisible goods can only help with the MMS approximation guarantee.
However, we show that this intuition no longer holds at the per-instance level.
In particular, we provide a problem instance with only indivisible goods, such that when a small amount of divisible goods is added to the instance, the MMS approximation guarantee of the instance \emph{strictly decreases}, i.e., while an $\alpha$-\mixMMS allocation exists in the original instance, no $\alpha$-\mixMMS allocation exists after adding the cake.

Next in Section~\ref{sec:enoughCakeHelp}, we focus on finding allocations with good MMS approximations with mixed types of goods.
More specifically, we show via a constructive algorithm that given any problem instance with mixed goods, there exists an \amms allocation, where the parameter $\alpha$, ranged between $1/2$ and $1$, is a monotonically increasing function of how agents value the divisible goods relative to their MMS values.
This means when agents have more divisible goods with them, one can achieve a better MMS approximation guarantee.
The idea of the algorithm is to repeatedly assign some agent a set of indivisible goods along with a piece of cake to reach the agent's \amms value, and then reduce the problem to a smaller size.
When the cake to be allocated is heterogeneous, the algorithm also makes use of a generalized fairness notion of \emph{weighted proportionality} to help allocate the cake.
On the computational front, we show polynomial-time approximation schemes for approximating the MMS value of an agent and for computing a $(1-\epsilon)$\amms allocation in a mixed goods problem instance.
% In particular, we give a PATS taht can always extract at least $1-\epsilon$ of an agent's MMS value.
These algorithms run in time polynomial in $n, m, L$ for any constant $\epsilon > 0$, where $n$ is the number of agents, $m$ is the number of indivisible goods, and $L$ is the input bit length.

Last, in Section~\ref{sec:relation}, we discuss the relation between MMS and the recently introduced notion of envy-freeness for mixed goods (EFM) in the mixed goods setting.
Generally speaking, neither MMS nor EFM imply the other.
We also provide a result showing what fraction of MMS can be implied by an EFM allocation.

\subsection{Further Related Work}
As mentioned earlier, proportionality fairness was first introduced seven decades ago in the seminal work by \citet{Steinhaus49} in the context of cake-cutting.
Since then, several efficient algorithms have been proposed~\citep{DubinsSp61,EvenPa84}, and a matching lower bound has been found by \citet{EdmondsPr11}.

% The overall number of queries for their method is $\Theta(n^2)$.
% Two decades later, \citet{EvenPa84} designed a more computationally efficient algorithm which also achieves the proportionality but takes at most $O(n \log n)$ queries.
% Furthermore, the Even-Paz algorithm is asymptotically optimal as any proportional cake-cutting algorithm requires $\Omega(n \log n)$ queries in the RW model~\citep{EdmondsPr11}.

Despite the fully fledged theory of cake cutting~\citep{BramsTa96,RobertsonWe98,Barbanel05}, it was not until recently attracted significant attention on fair division of indivisible goods.
Maximin share (MMS) fairness, often regarded as a generalization of proportionality to indivisible resource allocation, was first introduced by \citet{Budish11}.
\citet{KurokawaPrWa18} later showed that an MMS allocation may not always exist, but a $2/3$-MMS allocation always exists for any number of agents.
\citet{AmanatidisMaNi17} then devised an algorithm that computes a $(2/3 - \epsilon)$-MMS allocation with running time polynomial in the number of agents and goods.
The approximation guarantee for MMS was further improved to $3/4$ by \citet{GhodsiHaSe21} and the currently best-known ratio is $3/4 + 1/(12n)$ due to \citet{GargTa20}.

In addition to the works we mentioned above,
MMS allocations of indivisible resources have also been extensively studied in several other settings, including for agents with unequal entitlements~\citep{FarhadiGhHa19} or in different groups~\citep{Suksompong18}, for goods forming an undirected graph~\citep{BouveretCeKa17,LoncTr20}, for allocations under matroid constraints~\citep{GourvesMo19} or in conjunction with economic efficiency~\citep{IgarashiPe19}, as well as in the context of \emph{chore division}, where chores refer to negatively valued items~\citep{BarmanKr20,AzizRaSc17,AzizChLi19,AzizLiWu19,HuangLu19}.
\citet{CaragiannisKuMo19} introduced \emph{pairwise maximin share guarantee}, which is incomparable with MMS fairness.
It remains an open problem question whether a pairwise MMS allocation always exists~\citep{CaragiannisKuMo19}.
\citet{BarmanBiKr18} defined a stronger fairness notion than MMS, called \emph{groupwise maximin share guarantee (GMMS)}, and showed that GMMS allocations always exist in specific settings.

Besides proportionality and MMS fairness, another prominent fairness notion in resource allocation is \emph{envy-freeness (EF)}, which requires that each agent weakly prefers her own bundle to any other agent's bundle~\citep{Foley67}.
It follows from definition that envy-freeness implies proportionality when agents have additive valuations and all goods must be allocated.
An envy-free allocation of divisible goods for any number of agents always exists~\citep{Alon87} and can be computed via a discrete and bounded algorithm~\citep{AzizMa16}.
With indivisible goods, envy-freeness may not always be achievable.
This notion is then often relaxed to \emph{envy-freeness up to one good (EF1)}, which requires that any envy an agent has towards another agent can be eliminated by removing some good from the latter's bundle.
An EF1 allocation always exists for any number of agents and can be computed efficiently~\citep{LiptonMaMo04}.

Recently, \citet{BeiLiLi21} initiated the study of fair allocation of mixed divisible and indivisible goods.
They introduced the fairness notion of \emph{envy-freeness for mixed goods (EFM)}, which unifies EF and EF1 in the mixed good setting.
\citeauthor{BeiLiLi21} showed that an EFM allocation with mixed goods always exists for any number of agents and also investigated its computational aspects.
Later, \citet{BhaskarSrVa20} established an analogous existence result in a setting where \emph{undesirable} indivisible items and a cake are to be divided as well as a similar result in the flipped setting consisting of indivisible goods and \emph{undesirable} divisible items.

A related line of research incorporates money into the fair division of indivisible goods, with the consideration of finding envy-free allocations~\citep{AlkanDeGa91,Maskin87,Klijn00,MeertensPoRe02,HalpernSh19,BrustleDiNa20,CaragiannisIo20,Aziz21}.
In a recent work, \citet{HalpernSh19} bounded the amount of money needed to achieve envy-freeness for agents with additive valuations, assuming that the value of each agent for each good is at most $1$.
Their results were further improved by~\citet{BrustleDiNa20}.
On the other hand, \citet{CaragiannisIo20} studied the optimization problem of computing the minimum amount of subsidy needed to obtain envy-freeness given an allocation instance and showed both hardness and approximation results.
\citet{Aziz21} studied the problem of using monetary transfers to achieve envy-freeness and equitability simultaneously.\footnote{An allocation is said to satisfy \emph{equitability} if every agent gets the same utility.}
Our work and previous works indeed share lots of similarities, because money is just a special type of cake.
However, a major difference between our work and theirs is that we have different objectives.
This also makes our results and theirs incomparable.
More specifically, their works aim to \emph{determine} the amount of money needed to be added to the set of indivisible goods such that an envy-free allocation can be obtained.
However, their results do not specify what to do when there is not enough money as they required.
In our work, we focus on instances in which the amount of indivisible and divisible goods (or cake) are both \emph{fixed}, and regardless of whether the cake is enough to guarantee an MMS allocation, we aim to find a reasonably fair allocation with an MMS approximation guarantee.

Another closely related problem is \emph{rent division} (see, e.g.,~\citep{Su99,HaakeRaSu02,AbdulkadirogluSoUn04,Brams08,GalMaPrZi17,ArunachaleswaranBaRa19}). Its cardinal utility version can be viewed as a special case of the mixed setting where one wants to allocate (indivisible) rooms and the (divisible) rent among agents.
%An envy-free allocation always exists in rent division if we allow some agent can be paid to get a room.
However, in the mixed setting of fair division, the divisible goods (the rent) must be allocated and the agents are not allowed to use additional money to achieve more strict fairness condition.

\section{Preliminaries}\label{sec:preliminaries}
Denote by $N = \{1, 2, \dots, n\}$ the set of agents.
Let $M = \{1, 2, \ldots, m\}$ be the set of indivisible goods.
Each agent $i \in N$ has a non-negative utility $u_i(g)$ for each indivisible good $g \in M$.
We assume that each agent's utility for a set of indivisible goods is additive, that is, $u_i(M') = \sum_{g \in M'} u_i(g)$ for any $i \in N$ and $M' \subseteq M$.
Let $C = \{D_1, D_2, \ldots, D_\ell\}$ be the set of heterogeneous divisible goods.
We assume without loss of generality that each cake $D_i \in C$ is denoted by the interval $[(i-1)/\ell, i/\ell]$.
Thus the entire set of divisible goods is represented by one cake $C = [0, 1]$.
The agents' density functions over the cakes are assumed to be non-atomic.
This property allows us to view two consecutive intervals as disjoint if their intersection is a singleton.
A \emph{piece} of cake is a \emph{finite} union of subintervals of $[0, 1]$.
Each agent $i$ has a non-negative integrable density function $f_i$.
Given a piece of cake $S \subseteq [0, 1]$, agent $i$'s value over $S$ is then defined as $u_i(S) \coloneqq \int_{x \in S} f_{i}(x)\ dx$.
In this work, a resource allocation \emph{problem instance} $I = \langle N, M \cup C \rangle$ consists of a set of agents $N$ (together with their utility and density functions), a set of indivisible goods $M$, and a set of heterogeneous divisible goods or \emph{cakes} $C$.

Denote by $G = M \cup C$ the set of \emph{mixed} goods.
Let $\mathcal{M} = (M_1, M_2, \dots, M_n)$ be a partition of indivisible goods $M$ into $n$ bundles such that agent $i$ receives $M_i$.
Let $\mathcal{C} = (C_1, C_2, \dots, C_n)$ be a partition of the cake $C$ such that agent $i$ gets a piece of cake $C_i$.
An \emph{allocation} of mixed goods $G = M \cup C$ is defined as $\mathcal{A} = (A_1, A_2, \dots, A_n)$, where $A_i = M_i \cup C_i$ is allocated to agent $i$.
The \emph{utility} of agent $i$ in an allocation $\mathcal{A}$ is then $u_i(A_i) = u_i(M_i) + u_i(C_i)$.

We now define the fairness notions considered in this paper.
We focus on the \emph{maximin share fairness}, a generalization of the classic \emph{proportionality fairness}.

\begin{definition}[PROP]\label{def:prop}
An allocation $\mathcal{A}$ is said to satisfy \emph{proportionality (PROP)} if for each agent $i \in N$, $u_i(A_i) \geq u_i(G)/n$.
\end{definition}

\begin{definition}
Let $\Pi_k(G) = \{ \{P_1, P_2, \dots, P_k\} \mid P_i \cap P_j = \emptyset~\forall i \neq j \text{~and~} \bigcup_k P_k = G\}$ be the set of $k$-partitions of $G$.
Define the $k$-maximin share of agent $i$ as
\[
\mixMMS_i(k, G) = \max_{\mathcal{P} = (P_1, P_2, \dots, P_k) \in \Pi_k(G)}\min_{j \in [k]} u_i(P_j).
\]
The \emph{maximin share} of agent $i$ is $\mixMMS_i(n, G)$.
We say an \emph{MMS partition} for agent $i$ if this partition is in $\arg\max_{\mathcal{P} \in \Pi_n(G)}\min_{j \in [n]} u_i(P_j)$.
\end{definition}

For notational convenience, we will simply write $\mixMMS_i$ when parameters $n$ and $G$ are clear from the context.

\begin{definition}[\amms]
An allocation $\mathcal{A}$ of mixed goods $G$ is said to satisfy the \emph{$\alpha$-approximate maximin share fairness (\amms)}, for some $\alpha \in [0, 1]$, if for every agent $i \in N$,
\begin{equation*}
u_i(A_i) \geq \alpha \cdot \mixMMS_i(n, G).
\end{equation*}
\end{definition}

We say a $1$-MMS (or full-MMS) allocation satisfies the (full) maximin share fairness and write MMS as a shorthand for $1$-MMS.\footnote{While an MMS allocation may not exist in general, such an allocation always exists in the case of two agents.}
To slightly abuse the notation, we will also refer to an agent's maximin share as MMS.

\paragraph{Precision and Input Representation}
When discussing the computational aspects, it is necessary to specify the precision and representation of the input problem instance.
In this paper, we assume that $u_i(g)$'s for each $i \in N, g \in M$ and $u_i(C)$ for each $i \in N$ are all rational numbers, and the whole input can be represented in no more than $L$ bits.

\paragraph{Robertson-Webb Query Model}
We also adopt the Robertson-Webb (RW) query model to access agents' density functions for the cake.
In the RW model, an algorithm is allowed to ask each agent the following two types of queries:
\begin{description}
\item[Eval:] An evaluation query returns $u_i([x, y])$ of agent $i$ over interval $[x, y]$.
\item[Cut:] A cut query of $\beta$ for agent $i$ from point $x$ returns the leftmost point $y$ such that $u_i([x, y]) = \beta$.
\end{description}
In this paper, we assume that each query in the RW model takes unit time.

\section{MMS Approximation Guarantee}\label{sec:existence}
In this section, we examine how mixed goods affect the existence and approximation of MMS allocations.

\subsection{Worst Case MMS Approximation Guarantee}
An MMS allocation, while being an appealing solution concept, may not always exist in every problem instance with indivisible goods~\citep{KurokawaPrWa18}.
Therefore one has to resort to approximate MMS allocations.
Allocating mixed types of goods is a generalization of the indivisible good case, and hence suffers from the same issue.
We start by analyzing the worst-case MMS approximation guarantee for mixed good problem instances.

\begin{definition}
Given a mixed good problem instance $I$, let $\gamma(I)$ denote the maximum value of $\alpha$ such that the problem instance admits an \amms allocation.\footnote{The $\gamma(I)$ is defined to be the maximum value of $\alpha$ instead of the supremum.
This is because the density functions are non-atomic and the maximum $\alpha$ can always be achieved.}
We also call $\gamma(I)$ the \emph{MMS approximation guarantee} of problem instance $I$.
\end{definition}

We further define two constants
$$\gamma_M = \inf_{I = \langle N, M \cup C \rangle} \gamma(I) \quad \textrm{ and } \quad \gamma_I = \inf_{I = \langle N, M \rangle}\gamma(I).$$

In other words, $\gamma_M$ is the worst MMS approximation guarantee across all mixed goods problem instances, and $\gamma_I$ is the worst MMS approximation guarantee across all problem instances that contain only indivisible goods.
Previous works have shown that $\gamma_I < 1$~\citep{KurokawaPrWa18} and $\gamma_I \geq \frac34 + \frac{1}{12n}$~\citep{GargTa20}.

It is straightforward from definition that $\gamma_M \leq \gamma_I$.
In the following, our first result shows that $\gamma_M$ is also no less than $\gamma_I$.
This is proved via the following reduction theorem.

\begin{theorem}\label{thm:neg2}
Given any problem instance with mixed goods $I = \langle N, M \cup C \rangle$, there exists another problem instance $I' = \langle N, M' \rangle$ with only indivisible items $M'$ and the same set $N$ of agents, such that
\begin{itemize}
\item any allocation $\mathcal{A}'$ of $M'$ can be converted into another allocation $\mathcal{A}$ of $M \cup C$, such that $u_i(A_i) = u_i(A'_i)$ for each agent $i \in N$;
\item $\mixMMS_i(n, M \cup C) = \mixMMS_i(n, M')$ for each agent $i \in N$.
\end{itemize}
\end{theorem}
\begin{proof}
We first transform the mixed goods instance $I = \langle N, M \cup C \rangle$ into an instance $I' = \langle N, M' \rangle$ with only indivisible goods.
Consider an agent $i$ and an MMS partition $\mathcal{P}_i$ for this agent in $I$.
Assume that $\mathcal{P}_i$ divides cake $C$ into at most $n$ intervals with at most $n - 1$ cuts.
This assumption is without loss of generality because in an MMS partition it only matters how much value worth of cake is assigned to each bundle, but not their positions.
Then, by collecting all cuts of all $n$ MMS partitions $\mathcal{P}_1, \ldots, \mathcal{P}_n$ on $C$, they cut the cake into at most $n(n - 1) + 1$ pieces.
We can treat these pieces on $C$ as a set $M''$ of indivisible ``frozen pieces''.
Together with $M$, we now have $M' = M'' \cup M$.

Given any allocation $\mathcal{A}'$ of $M'$, we can easily convert it into an allocation $\mathcal{A}$ of $G$ by transforming those ``frozen cake pieces'' back to normal cake pieces.
This also gives $u_i(A_i) = u_i(A'_i)$ for each agent $i \in N$, which proves the first part of Theorem~\ref{thm:neg2}.

Last, it is clear that every agent can have the same MMS partition in $I'$ as that in $I$, because the cuts do not affect their MMS partitions.
This implies that $\mixMMS_i(n, M') \geq \mixMMS_i(n, M \cup C)$ for each agent $i \in N$.
On the other hand, the first part of this theorem also implies $\mixMMS_i(n, M \cup C) \geq \mixMMS_i(n, M')$.
Hence we have $\mixMMS_i(n, M \cup C) = \mixMMS_i(n, M')$ for each agent $i \in N$.\qed
\end{proof}

We note that this reduction is not computationally efficient as it requires being able to compute the MMS values. Moreover, Theorem~\ref{thm:neg2} directly implies the following result.

\begin{corollary}\label{cor:gi=gm}
% The universal lower bound of the approximation ratio $\alpha$ in all mixed goods instances will no worse than the lower bound in all indivisible goods cases.
$\gamma_I = \gamma_M$.
\end{corollary}

In other words, having mixed types of goods does not affect the worst-case MMS approximation guarantee across all problem instances.
As another corollary, this also means that if there exists a universal $\beta$-MMS algorithm for indivisible goods for some $\beta$, it immediately implies that every problem instance of mixed goods also admits a $\beta$-MMS allocation.
We will discuss more on the algorithmic implication of this result in Section~\ref{sec:enoughCakeHelp}.

\subsection{Cake Does Not Always Help}
Note that the equation in Corollary~\ref{cor:gi=gm} is about the worst-case MMS approximation guarantee across all problem instances.
Next we show that such equivalence may not hold on a per-instance level.
In particular, we will demonstrate via an example that sometimes, adding some divisible resources to some problem instance $I$ may hurt its MMS approximation guarantee value $\gamma(I)$.\footnote{In the indivisible setting, the corresponding result of adding an item may lower the best MMS guarantee for a problem instance is easy to get.}

\begin{theorem}\label{thm:uselessCake}
For any $n \geq 6$, there exist some agent set $N$, indivisible goods $M$, and divisible goods $C$, such that
\[
\gamma(\langle N, M \rangle) > \gamma(\langle N, M \cup C \rangle).
\]
In other words, adding some divisible goods to the set of resources may decrease the MMS approximation guarantee of this problem instance in some cases.
\end{theorem}

Before showing the detailed proof, we first explain the intuition of the theorem proof.
We want to find a problem instance $I = \langle N, M \rangle$ such that $\gamma (\langle N, M \rangle) < 1$, and the instance should have the following properties.

Fix an agent $i$.
In her MMS partition, the least valued bundle is unique, i.e., the value of the least valued bundle is \emph{strictly} less than that of the second least valued bundle.
If this is the case, then given a cake $C$ with a small enough value $\epsilon$, the new MMS value $\mixMMS_i(n, M \cup C)$ should be exactly $\mixMMS_i(n, M) + \epsilon$.
Now suppose that in the instance $I$, all of the agents have this property.
This means that every agent's MMS value will increase by $\epsilon$ when we add a cake $C$ of a small enough value $\epsilon$ to the instance $I$.
The second required property of $I$ is that in any $\gamma (\langle N, M \rangle)$-MMS allocation, there are at least two agents that receive exactly $\gamma (\langle N, M \rangle)$ times their MMS values.

With these two properties, the actual cake $C$ will not be enough for distributing to all of the agents while clinging to a large enough MMS approximation ratio $\gamma( \langle N, M \cup C \rangle )$.
In other words, with the cake $C$ added, the new MMS ratio $\gamma( \langle N, M \cup C \rangle)$ will decrease, comparing to $\gamma (\langle N, M \rangle)$.

Finally, the counterexample used to show the non-existence of MMS allocation in~\citep{KurokawaPrWa18} can be utilized to construct the instance $I$ that satisfies all above mentioned properties.
By utilizing their construction, our argument requires at least six agents.
The full proof can be found in the following.
\begin{proof}
Our counterexample will utilize the following lemma from~\cite{KurokawaPrWa18}, which is also used for showing the non-existence of full MMS allocation with indivisible goods.

\begin{lemma}[Base of counterexample~\citep{KurokawaPrWa18}]\label{lem:base}
    For any $n \geq 6$, there exists an $n \times n$ matrix $M$, satisfying the following properties:
    \begin{enumerate}
        \item All entries are non-negative (i.e., $\forall i, j \colon M_{i,j} \geq 0$).
        \item All entries of the last row and column, and the first entry in the first row, are positive (i.e., $\forall i \colon M_{i, n}, M_{n, i} > 0$ and $M_{1, 1} > 0$).
        \item All rows and columns sum to $1$ (i.e., $M \vec{1} = M^\top \vec{1} = \vec{1}$).
        \item Define $M^{+}$ as the set of all positive entries in $M$.
        Then if we wish to partition $M^{+}$ into $n$ subsets that sum to exactly $1$, then our partition must correspond to either the rows of $M$ or the columns of $M$.
    \end{enumerate}
    %(1) Allentriesarenonnegative(i.e.,∀i,j:Mi,j ≥0).
    %(2) All entries of the last row and column are positive (i.e., ∀i : Mi,n,Mn,i > 0).
    %(3) Allrowsandcolumnssumto1(i.e.,M⃗1=MT⃗1=⃗1).
    %(4) Define M+ as the set of all positive entries in M. Then if we wish to partition M+ into n subsets that sum to exactly 1, then our partition must correspond to the rows of M or the columns of M.
\end{lemma}

Then construct two $n \times n$ matrices $P^+$ and $P^-$.
Let $P^+_{1, 1} = P^-_{1, 1} = -\epsilon$, $P^+_{n, 1} = P^-_{1, n} = -\epsilon$, $P^+_{n, n} = P^-_{n, n} = (2n-3)\epsilon$, and $P^+_{n, i} = P^-_{i, n} = -2\epsilon$ for $2 \leq i \leq n - 1$.
Take $n = 6$ as an example, we show below the construction of matrices $P^+$ and $P^-$:
\[
P^{+} = \left[
\begin{matrix}
-\epsilon&0&0&0&0&0\\
0&0&0&0&0&0\\
0&0&0&0&0&0\\
0&0&0&0&0&0\\
0&0&0&0&0&0\\
-\epsilon&-2\epsilon&-2\epsilon&-2\epsilon&-2\epsilon&9\epsilon\\
\end{matrix}
\right]
\]
\[
P^{-} = \left[
\begin{matrix}
-\epsilon&0&0&0&0&-\epsilon\\
0&0&0&0&0&-2\epsilon\\
0&0&0&0&0&-2\epsilon\\
0&0&0&0&0&-2\epsilon\\
0&0&0&0&0&-2\epsilon\\
0&0&0&0&0&9\epsilon\\
\end{matrix}
\right]
\]

Consider a matrix $M$ satisfying all properties listed in Lemma~\ref{lem:base}.
By setting a properly small value $\epsilon$, we can always make sure that every entry of $M + P^+$ and $M + P^-$ is non-negative.
In the following, we will treat each entry as an indivisible good.
We next divide $N$ into two disjoint subsets.
One contains $\left\lfloor \frac{n}{2} \right\rfloor$ agents, denoted by $N^+$.
The other contains the rest of the agents, denoted by $N^-$.
We let each agent $i \in N^+$ take the values of $n^2$ items as in matrix $M + P^+$, and each agent $i \in N^-$ take the values of $n^2$ items as in matrix $M + P^-$. We call this problem instance $I$.
One can check that in this instance $I$, the MMS value for each agent $i$ is $1 - \epsilon$.

According to the fourth property in Lemma~\ref{lem:base}, there are only two ways to distribute these items into $n$ bundles such that each bundle has value close to $1$: either the rows of $M$ or the columns of $M$.
It means that once fixing the partition (by rows or by columns), half of the agents may receive bundles with value $1-2\epsilon$ or the bundle with value more than $1$.
% In other words, our allocation must use one of these two methods to divide all items into $n$ bundles.
In each of these two partitions, due to $n \geq 6$, we can always find at least $2$ agents who value their bundles exactly $1 - 2\epsilon$.
For example, there are two such agents from $N^+$ if the partition is $n$ columns, or two from $N^-$ if the partition is $n$ rows.
% Let $S$ be the set of such agents.
In particular, this means the MMS approximation guarantee $\gamma(I)$ of this instance $I$ is $\frac{1 - 2\epsilon}{1 - \epsilon}$.

Suppose now we add a homogeneous cake to this problem instance $I$. This cake has value $\epsilon$ for each agent.
Every agent's MMS value will now increase from $1 - \epsilon$ to $1$.
However, in any allocation, there will still be at least two agents whose values for the indivisible goods are no more than $1-2\epsilon$. Then the best possible way to distribute the cake is to allocate it only to those agents, which means at least one such agent will receive a bundle of value at most $1 - 2\epsilon + \frac{\epsilon}{2} = 1 - \frac{3\epsilon}{2}$.
Thus, in this case, the MMS approximation ratio of such agent will be no more than $1 - \frac{3\epsilon}{2}$, which is strictly smaller than $\frac{1 - 2\epsilon}{1 - \epsilon}$ when $\epsilon < 1/4$.\qed
\end{proof}

\section{Algorithms for Computing Approximate MMS Allocations}\label{sec:enoughCakeHelp}
The previous section investigates MMS approximation guarantee, which is the \emph{best possible} MMS approximation of a problem instance.
In this section, our goal is to design algorithms that could compute allocations with \emph{good} MMS approximation ratios in mixed goods problem instances.
We hope such an algorithm can be flexible, in the sense that when the problem instance contains only indivisible goods, the MMS approximation of the output allocation should match or be close to the previously best-known approximation ratio for indivisible goods; on the other hand, when the resources contain enough divisible goods, the indivisible goods would become negligible, and our algorithm should be able to produce an allocation that gives each agent their full MMS value.

As the main result of this section, in the following we present such an algorithm.
We will show that the algorithm will always produce an \amms allocation in the mixed goods setting, where $\alpha$ is a monotonically increasing function of how agents value the divisible goods relative to their MMS values and ranges between $1/2$ and $1$.

\begin{theorem}\label{thm:heteroCake}
Given any mixed good problem instance $\langle N, M \cup C \rangle$, an \amms allocation always exists, where
$$\alpha = \min\left\{1, \frac12 + \min_{i \in N}\left\{\frac{u_i(C)}{2(n-1) \cdot \mixMMS_i}\right\}\right\}.$$
Furthermore, for any constant $\epsilon > 0$, we can compute a ratio $\alpha'$ and an allocation $\mathcal{A}$ in time polynomial in $n, m, L$ such that:
\begin{enumerate}
\item $\alpha' \geq \alpha$, and
\item the allocation $\mathcal{A}$ is $(1-\epsilon)\alpha'$-MMS.
\end{enumerate}
Here $n$ is the number of agents, $m$ is the number of items, and $L$ is the total bit length of all input parameters.
\end{theorem}

Theorem~\ref{thm:heteroCake} has several implications.
% \shengxin{For example, given that an indivisible good with a sufficiently small value $\epsilon > 0$ and a cake with value $u_i(C) = 1 - \epsilon$ to each agent $i \in N$, the MMS value of each agent $i$ is $\frac{1}{n}$. Theorem~\ref{thm:heteroCake} thus implies that $\alpha$ is at least 1 which is better than the currently best-known result with indivisible goods $\frac34 + \frac{1}{12n}$ due to~\citet{GargTa20}.}
For example, when every agent $i$ has $u_i(C) \geq (n/2) \mixMMS_i$, Theorem~\ref{thm:heteroCake} implies the existence of an $\alpha$-MMS allocation with $\alpha$ better than the currently best-known approximation ratio of $\frac{3}{4} + \frac{1}{12n}$ with indivisible goods due to~\citet{GargTa20}.
In addition, the following corollary shows the amount of divisible goods needed to ensure that the instance admits a full-MMS allocation.

\begin{corollary}
Given a mixed good problem instance $I = \langle N, M \cup C \rangle$, if $u_i(C) \geq (n-1)\mixMMS_i$ holds for each agent $i \in N$, then an MMS allocation is guaranteed to exist.
\end{corollary}

This means that even with the presence of indivisible items, as long as there are enough divisible goods, a full-MMS allocation can always be found.
However, we note that this corollary should not be interpreted as that this is the \emph{least} amount of divisible goods required.
For example, \citet{HalpernSh19} and \citet{BrustleDiNa20} studied the allocation of indivisible goods and a very special type of divisible goods, money.
They bounded the amount of money needed for an indivisible goods instance to have an envy-free allocation, assuming that the value of each agent for each good is at most $1$.
Although an envy-free allocation is also a full-MMS allocation, their results and this corollary are incomparable because we have different objectives, and it is not our goal to find the minimum amount of cake needed to ensure an MMS allocation.

The remaining of this section is dedicated to the proof of Theorem~\ref{thm:heteroCake}.
The proof consists of the following steps.
\begin{description}
\item [Section~\ref{ssec:homoCake}:] We first focus on a restricted case in which the cake to be allocated is \emph{homogeneous} to every agent.
We show via a constructive but not necessarily polynomial time algorithm that an \amms allocation always exists in this setting.

\item [Section~\ref{ssec:heteroCake}:] Next we generalize the above algorithm to the general case with heterogeneous cake, using the concept of \emph{weighted proportionality} in cake-cutting.

\item [Section~\ref{ssec:computation}:] We discuss how to convert the algorithm into a polynomial-time algorithm at the cost of a small loss in the MMS approximation ratio.
\end{description}

We also discuss how to further improve the approximation ratio $\alpha$ in Section~\ref{sec:improvement}.

\subsection{Homogeneous Cake}\label{ssec:homoCake}
We begin with a special case where the cake to be allocated is \emph{homogeneous}, meaning that each agent values all pieces of equal size the same.
In other words, the value of a piece of cake to each agent depends only on the length of the piece.
We refer to the homogeneous cake as $\hat{C}$.
Formally, given a piece of homogeneous cake $S \subseteq [0, 1]$, each agent $i$'s value over $S$ is then defined as $u_i(S) \coloneqq (\sum_{[a,b] \in S} (b-a)) u_i(\hat{C})$.
% We also note that within this section, the mixed goods $G$ refer to the indivisible goods $M$ and the homogeneous cake $\hat{C}$.

\subsubsection{The Algorithm}
\begin{algorithm}[t]
\caption{$\homoAlg(\langle N, M \cup \hat{C} \rangle)$}
\label{alg:mixedMMS-homoCake}
\DontPrintSemicolon

\KwIn{Agents $N$, indivisible goods $M$ and a homogeneous cake $\hat{C}$, utility and density functions.}

Compute $\mixMMS_i$, for each $i \in N$.\; \label{ALGHomo:mms}
$\alpha \leftarrow \min\left\{1, \frac12 + \min_{i \in N}\left\{\frac{u_i(\hat{C})}{2(n-1) \cdot \mixMMS_i}\right\}\right\}$\; \label{ALGHomo:alpha}
$A_1, A_2, \dots, A_n \leftarrow \emptyset$\;

\BlankLine
\tcp{Phase 1: allocate big goods.}
\While{$\exists i \in N, g \in M$ such that $u_i(g) \geq \alpha \cdot \mixMMS_i$}{ \label{ALGHomo:reductionBegin}
	$A_i \leftarrow \{g\}$ \tcp*{arbitrary tie-breaking}
	$N \leftarrow N \setminus \{i\}$, $M \leftarrow M \setminus \{g\}$\;
} \label{ALGHomo:reductionEnd}

\BlankLine
\tcp{Phase 2: allocate small goods.}
\While{$|N| \geq 2$}{ \label{ALGHomo:bagfillBegin}
	$B \leftarrow \emptyset$\;
	Add one indivisible good at a time to $B$ until $u_j(B) \geq (1-\alpha) \cdot \mixMMS_j$ for some agent $j$ or $B = M$.\; \label{ALGHomo:indBundle}
	Suppose $\hat{C} = [a, b]$. For each $i \in N$, let $x_i$ be the leftmost point with $u_i(B \cup [a, x_i]) \geq \alpha \cdot \mixMMS_i$.\; \label{ALGHomo:markPoint}
	$i^* \leftarrow \arg\min_{i \in N} x_i$ \tcp*{arbitrary tie-breaking} \label{ALGHomo:iStar}
	$A_{i^*} \leftarrow B \cup [a, x_{i^*}]$\; \label{ALGHomo:ind+cakeAlloc}
	$N \leftarrow N \setminus \{i^*\}$, $M \leftarrow M \setminus B$, $\hat{C} \leftarrow \hat{C} \setminus [a, x_{i^*}]$\;
} \label{ALGHomo:bagfillEnd}
Give all remaining goods to the last agent.\; \label{ALGHomo:lastAgent}

\Return $(A_1, A_2, \dots, A_n)$
\end{algorithm}

The complete algorithm to compute an \amms allocation is shown in Algorithm~\ref{alg:mixedMMS-homoCake}.
Our algorithm is in spirit similar to the algorithm in~\citep{GhodsiHaSe21}.
After initialization, the algorithm can be decomposed into two phases as follows:

\begin{itemize}
\item Phase 1: allocate big goods (lines~\ref{ALGHomo:reductionBegin}-\ref{ALGHomo:reductionEnd}). Algorithm~\ref{alg:mixedMMS-homoCake} repeatedly allocates some agent a single indivisible good which has a value at least $\alpha$ times this agent's MMS value.
Then, both the agent and the allocated good are removed from all further considerations.

\item Phase 2: allocate small goods (lines~\ref{ALGHomo:bagfillBegin}-\ref{ALGHomo:bagfillEnd}). This phase executes in rounds.
In each round, Algorithm~\ref{alg:mixedMMS-homoCake} chooses an agent $i^*$ and allocates some indivisible goods $B$ (formed at line~\ref{ALGHomo:indBundle}) along with a piece of cake $[a, x_{i^*}]$ to agent $i^*$ (line~\ref{ALGHomo:ind+cakeAlloc}).
Then again, both the agent and her goods are removed from the instance.
\end{itemize}

\subsubsection{The Analysis}
% Our result for the \amms allocation of indivisible goods and homogeneous cake is shown as follows.

% \begin{theorem}\label{thm:homoCake}
% 	When the cake is homogeneous, an \amms allocation always exists for any number of agents and can be found by Algorithm~\ref{alg:mixedMMS-homoCake}.
% \end{theorem}
% To show that Algorithm~\ref{alg:mixedMMS-homoCake} correctly outputs an \amms allocation when the cake is homogeneous, it suffices to prove the following two claims:
Algorithm~\ref{alg:mixedMMS-homoCake} consists of two phases.
We analyze each of them separately.

\paragraph{Phase 1: Allocate Big Goods}
First, when goods are all indivisible, it follows from Lemma~1 of \citet{BouveretLe16} that allocating a single good to an agent does not decrease the MMS values of other agents.
Here we show that this result holds in the mixed goods setting as well.

\begin{lemma}[Monotonicity property]\footnote{We adopt the name ``monotonicity property'' from \citet{AmanatidisMaNi17}.}\label{lem:monotonicity}
Given an instance $\langle N, G = M \cup C \rangle$, for any agent $i \in N$ and any indivisible good $g \in M$, it holds that
\[
\mixMMS_i(n-1, G \setminus \{g\}) \geq \mixMMS_i(n, G).
\]
\end{lemma}
\begin{proof}
Removing a single indivisible good in an MMS partition of agent $i$ affects exactly one bundle and each of the remaining $n - 1$ bundles has value at least $\mixMMS_i(n, G)$.
Therefore, we have $\mixMMS_i(n-1, G \setminus \{g\}) \geq \mixMMS_i(n, G)$.\qed
\end{proof}

Denote by $N_1$ the set of remaining agents and $G_1$ the set of unallocated goods just before Phase 2 is executed.
Let $n_1 = |N_1|$.
Applying the monotonicity property (Lemma~\ref{lem:monotonicity}) $n-n_1$ times, we have that for each agent $i \in N_1$, $\mixMMS_i(n_1, G_1) \geq \mixMMS_i(n, G)$.
In addition, each agent $i$ who leaves the system in this phase receives an item of value at least $\alpha \cdot \mixMMS_i$.
This implies that Phase 1 will not affect the correctness and termination of Algorithm~\ref{alg:mixedMMS-homoCake}.
It simply adds the property that in Phase 2, each remaining agent $i$ will value each of the remaining indivisible goods less than $\alpha \cdot \mixMMS_i$.

\paragraph{Phase 2: Allocate Small Goods}
In this phase, at each round, for the agent $i^*$ selected at line~\ref{ALGHomo:iStar}, we show that it satisfies two properties:
\begin{enumerate}[label=(\arabic*)]
\item $u_{i^*}(A_{i^*}) \geq \alpha\cdot\mixMMS_{i^*}$;
\item For each agent $j$ remaining in $N$, $u_j(A_{i^*}) \leq \mixMMS_j$.
\end{enumerate}
(1) is straightforward by the way each $x_i$ is computed at line~\ref{ALGHomo:markPoint}.
To show (2) is true, we remark that no single good is valued more than $\alpha \cdot \mixMMS_i$ for any agent $i$.
Therefore, the set $B$ selected at line~\ref{ALGHomo:indBundle} must satisfy $u_j(B) \leq \mixMMS_j$ for all $j \in N$.
In line~\ref{ALGHomo:markPoint}, each agent cuts a piece of cake such that the sum of her value for $B$ and this piece of cake is at least $\alpha$ fraction of her maximin share.
Because $\alpha \leq 1$ and the cake is divisible, after line~\ref{ALGHomo:markPoint}, it continues to satisfy that $u_j(B \cup [a, x_j]) \leq \mixMMS_j$ for each $j \in N$.
Then, because $i^*$ is selected such that $x_{i^*}$ is the smallest value, one would have $u_j(A_{i^*} = B \cup [a, x_{i^*}]) \leq u_j(B \cup [a, x_j]) \leq \mixMMS_j$ for each agent $j \in N$.

In particular, property (2) ensures that the last agent at line~\ref{ALGHomo:lastAgent} is still left with enough goods to reach her maximin share.
Therefore, every agent $i$ will receive value at least $\alpha\cdot\mixMMS_i$ after the two phases. It only remains to show that the cake $\hat{C}$ is enough to be allocated throughout the process.

\begin{lemma}\label{lem:cakeIsEnough}
Cake $\hat{C}$ is enough to be allocated in Algorithm~\ref{alg:mixedMMS-homoCake}.
In other words, $x_i$ for each agent $i \in N$ at line~\ref{ALGHomo:markPoint} is always well defined in each round.
\end{lemma}

% \begin{proof}[Proof sketch]
% Line~\ref{ALGHomo:alpha} is equivalent to $u_i(\hat{C}) \geq (n-1) \cdot (2\alpha - 1) \cdot \mixMMS_i, \forall i \in N$.
% As a result, each agent $i$ has value at least $(2\alpha - 1) \cdot \mixMMS_i$ for a $\frac{1}{n-1}$ fraction of the entire cake $\hat{C}$.
% In each round of Phase 2, a $\frac{1}{n-1}$ fraction of $\hat{C}$ is enough to make some agent achieve $\alpha$ times her MMS.
% Moreover, there are at most $n-1$ rounds of Phase 2.
% In the case that line~\ref{ALGHomo:indBundle} is not executed, each $i \in K$, where $K$ denotes the set of $k$ remaining agents, has value at least $k \cdot \mixMMS_i$ for the remaining cake $C'$, so $C'$ is enough to guarantee $\alpha \cdot \mixMMS_i, \forall i \in K$.
% \end{proof}

\begin{proof}
Line~\ref{ALGHomo:alpha} in Algorithm~\ref{alg:mixedMMS-homoCake} indicates that for each agent $i \in N$, $u_i(\hat{C}) \geq (n-1) \cdot (2\alpha - 1) \cdot \mixMMS_i$.
As a result, each agent $i$ has value at least $(2\alpha - 1) \cdot \mixMMS_i$ for a $\frac{1}{n-1}$ fraction of the entire cake $\hat{C}$.
It is also clear that Phase 2 has been executed at most $n-1$ times during the algorithm run.
That is to say the action of cutting a piece of $\hat{C}$ and allocating this piece to an agent is performed at most $n-1$ times.

Based on whether there exists some agent who has value at least $1-\alpha$ times her MMS for goods in $B$ (line~\ref{ALGHomo:indBundle}), we distinguish two cases.

\begin{itemize}
\item \emph{Line~\ref{ALGHomo:indBundle}: there exists some agent $j$ with $u_j(B) \geq (1-\alpha) \cdot \mixMMS_j$.}
As mentioned earlier, a $\frac{1}{n-1}$ fraction of $\hat{C}$ is worth at least $(2\alpha-1) \cdot \mixMMS_j$.
Thus it together with $B$ is enough to give agent $j$ a value of at least $\alpha \cdot \mixMMS_j$.
This means at line~\ref{ALGHomo:markPoint}, the length of $[a, x_j]$ is no more than $\frac{1}{n-1}$.
Moreover, Algorithm~\ref{alg:mixedMMS-homoCake} chooses the agent who claims the smallest piece of cake as agent $i^*$ at line~\ref{ALGHomo:iStar}, which means the length of $[a, x_{i^*}]$ is again no more than $\frac{1}{n-1}$.
Combining the fact that Phase 2 executes at most $n-1$ times, if this case holds every time, the cake will be enough.

\item \emph{Line~\ref{ALGHomo:indBundle}: $u_j(B) < (1-\alpha) \cdot \mixMMS_j$ for each agent $j$.}
In this case, $B$ is set to be $M$ at line~\ref{ALGHomo:indBundle}.
Note that after the first time of such case, $M$ will become empty, and the agents left will divide only the cake for the remaining rounds.
Let $k$ be the number of the remaining agents when $M$ becomes empty.
By property (2) that we showed above, we know the remaining cake is valued at least $k \cdot \mixMMS_i$ for each remaining agent $i$.
Thus it is enough for each agent $i$ to receive a piece with value at least $\alpha\cdot\mixMMS_i$.
\end{itemize}
The lemma thus follows.\qed
\end{proof}

Combining everything together, we conclude that Algorithm~\ref{alg:mixedMMS-homoCake} is a correct algorithm that always outputs an \amms allocation.

\subsection{Heterogeneous Cake}\label{ssec:heteroCake}
We now show how to extend algorithm~\ref{alg:mixedMMS-homoCake} to the general setting with a heterogeneous cake $C$.
The new algorithm follows a very simple idea as follows.
First we replace cake $C$ with a homogeneous cake $\hat{C}$ such that $u_i(\hat{C}) = u_i(C)$ for each agent $i$, and allocate resources $M$ and $\hat{C}$ to all agents using Algorithm~\ref{alg:mixedMMS-homoCake}.
Let $\hat{C}_i$ be the piece allocated to agent $i$.
Note that since $\hat{C}$ is homogeneous, only the length of $\hat{C}_i$ matters, which we denote as $w_i$.
Since $\hat{C}$ has total length 1, $w_i$ also represents the fraction of the cake $\hat{C}$ allocated to agent $i$.
Next, we view $w_i$ as the \emph{entitlement} (or \emph{weight}) of agent $i$ to the real cake $C$, and obtain the actual allocation of cake $C$ via a procedure known as the \emph{weighted proportional allocation}.

\paragraph{Weighted Proportional Cake Cutting}
This concept generalizes the proportional cake-cutting to the weighted case.
Formally, assume that every agent $i \in N$ is assigned a non-negative \emph{weight} $w_i$, such that $\sum_{i \in N} w_i = 1$.
We call the vector of weights $\mathbf{w} = (w_1, w_2, \dots, w_n)$ a \emph{weight profile}.

\begin{definition}[WPR]
Given a weight profile $\mathbf{w}$, an allocation $\mathcal{C} = (C_1, C_2, \dots, C_n)$ of cake $C$ is said to satisfy \emph{weighted proportionality (WPR)} if for every agent $i \in N$, $u_i(C_i) \geq w_i \cdot u_i(C)$.
\end{definition}

A weighted proportional allocation of cake gives each agent at least her entitled fraction of the entire cake from her own perspective.
The proportionality fairness (Definition~\ref{def:prop}) is a special case of WPR with weight profile $\mathbf{w} = (1/n, 1/n, \dots, 1/n)$.
With any set of agents and any weight profile, a weighted proportional allocation always exists~\citep{CsehFl20}.
In the following, we will assume that our algorithm is equipped with a protocol $\weightedPROP(N, C, \mathbf{w})$ that could return us a weighted proportional allocation of cake $C$, among the set of agent $N$ with weight profile $\mathbf{w}$.

\begin{algorithm}[t]
\caption{The Mixed MMS Algorithm}
\label{alg:mixedMMS}
\DontPrintSemicolon

\KwIn{Agents $N$, indivisible goods $M$ and cake $C$, utility and density functions.}

Let $\hat{C} = [0, 1]$ be a homogeneous cake with $u_i(\hat{C}) = u_i(C)$ for each agent $i \in N$.\; \label{ALGHetero:homoCake}
$(M_1 \cup \hat{C}_1, M_2 \cup \hat{C}_2, \dots, M_n \cup \hat{C}_n) \leftarrow \homoAlg(\langle N, M \cup \hat{C} \rangle)$\; \label{ALGHetero:callHomoAlg}
For each $i \in N$, let $w_i \leftarrow u_i(\hat{C}_i) / u_i(C)$ if $u_i(C) > 0$; $w_i \leftarrow 0$ otherwise.\;
$(C_1, C_2, \dots, C_n) \leftarrow \weightedPROP(N, C, \mathbf{w} = (w_1, \ldots, w_n))$ \tcp*{allocate cake $C$} \label{ALGHetero:cakeAlloc}
\Return $(M_1 \cup C_1, M_2 \cup C_2, \dots, M_n \cup C_n)$\;
\end{algorithm}

The complete algorithm to compute an \amms allocation of mixed goods for any number of agents is shown in Algorithm~\ref{alg:mixedMMS}.
To show that this algorithm can find an \amms allocation with mixed goods that contain a heterogeneous cake, it suffices to prove the following two simple facts.
\begin{enumerate}
\item $\mixMMS_i(n, M \cup C) = \mixMMS_i(n, M \cup \hat{C})$.
This is obvious because both $C$ and $\hat{C}$ are divisible with $u_i(C) = u_i(\hat{C})$.
Only changing the density of a cake will not affect the MMS value of any agent.
\item $u_i(C_i) \geq u_i(\hat{C}_i)$.
This is because by weighted proportionality, we have $$u_i(C_i) \geq w_i \cdot u_i(C) = w_i \cdot u_i(\hat{C}) = u_i(\hat{C}_i).$$
\end{enumerate}

\subsection{Computation}\label{ssec:computation}
We investigate the computational issues in finding an \amms allocation in this part.
Note that Algorithm~\ref{alg:mixedMMS} is not a polynomial-time algorithm unless P=NP.
This is because it requires the knowledge of every agent's MMS value, which is NP-hard to compute even with only indivisible resources~\citep{KurokawaPrWa18}.

To obtain a polynomial-time approximation algorithm, we first show how to approximate the MMS value of an agent with mixed goods, then focus on obtaining an approximate \amms allocation.

\subsubsection{Approximate MMS Value with Mixed Goods}
When goods are indivisible, \citet{Woeginger97} showed a polynomial-time approximation scheme (PTAS) to approximately compute the MMS value of an agent.
More specifically, given any constant $\delta > 0$ and any agent, we can partition the indivisible goods into $n$ bundles in polynomial time, such that each bundle is worth at least $1 - \delta$ of that agent's MMS value.
By utilizing this PTAS from~\citet{Woeginger97}, here we present a new PTAS to approximate MMS values for mixed goods.
% Given an instance $\mathcal{I}$, we discretize the cake into tiny indivisible intervals and denote by $\mathcal{I}'$ the discretized instance.
% In the following, we establish the connection between the maximin share of an agent in $\mathcal{I}$ and that in $\mathcal{I}'$.

\begin{lemma}\label{lem:approxMMS}
Given any mixed goods instance $I = \langle N, M \cup C \rangle$ and constant $\epsilon > 0$, for any agent $i \in N$, one can compute a partition $(P_1, P_2, \dots, P_n)$ of $M \cup C$ in polynomial time, such that $\min_{j \in N}u_i(P_j) \geq (1-\epsilon) \cdot \mixMMS_i(n, M \cup C)$.
% in time polynomial in $n, m$ and exponential in $1/\epsilon$.
\end{lemma}
\begin{proof}
Let agent $i$ cut the cake $C$ into $\left\lceil \frac{2n}{\epsilon} \right\rceil$ disjoint intervals worth at most $\frac{\epsilon \cdot u_i(C)}{2n}$ each to this agent.
Denote by $\tilde{C}$ the collection of these discretized, indivisible intervals.
The new discretized instance is then denoted by $I' = \langle N, M \cup \tilde{C} \rangle$.
This is a problem instance with only indivisible goods.

We first claim that
\begin{align*}
\mixMMS_i(n, M \cup C) &\geq \mixMMS_i(n, M \cup \tilde{C}) \geq \left(1-\frac{\epsilon}{2}\right) \cdot \mixMMS_i(n, M \cup C).
\end{align*}
The first inequality holds trivially by definition.
We proceed to show the second.
Consider an MMS partition $\mathcal{T}$ of $I$ for agent $i$.
We construct a partition $\mathcal{T}'$ of $I'$ as follows.
First, let the partition of its original indivisible goods $M$ be exactly the same as that in $\mathcal{T}$.
We then distribute the intervals in $\tilde{C}$ into these $n$ bundles.
For any bundle whose value is less than $\left(1-\frac{\epsilon}{2}\right) \cdot \mixMMS_i(n, M \cup C)$ to agent $i$, add one interval at a time to this bundle until agent $i$'s value for this bundle falls in $\left[\left(1-\frac{\epsilon}{2}\right) \mixMMS_i(n, M \cup C), \mixMMS_i(n, M \cup C)\right]$.
This is possible because $\mixMMS_i(n, M \cup C) \geq u_i(C)/n$ and each interval is worth at most $\frac{\epsilon \cdot u_i(C)}{2n} \leq \frac{\epsilon}{2} \cdot \mixMMS_i(n, M \cup C)$.
Also $\tilde{C}$ will have enough pieces for these allocations because in $\mathcal{T}$, each bundle is worth at least $\mixMMS_i(n, M \cup C)$ to agent $i$.
Repeat this procedure for all bundles.
Finally, distribute any remaining intervals to any of these bundles arbitrarily.
Let the resulting partition be $\mathcal{T'}$.

By the end of these procedures, each bundle in $\mathcal{T'}$ is worth at least $(1-\frac{\epsilon}{2}) \cdot \mixMMS_i(n, M \cup C)$.
Then by the definition of MMS, the second inequality holds.
We remark that these steps are not actually implemented in our algorithm.
They are only used to demonstrate the difference of MMS values for the two instances.

Now, because $I'$ is a problem instance with only indivisible goods, we can compute a partition $(P_1, P_2, \dots, P_n)$ such that $\min_{j \in N}u_i(P_j) \geq (1-\frac{\epsilon}{2}) \cdot \mixMMS_i(n, M \cup \tilde{C})$ via the PTAS from~\citep{Woeginger97} with $\delta = \epsilon/2$.
It then holds that
\begin{align*}
\min_{j \in N}u_i(P_j) &\geq \left(1-\frac{\epsilon}{2}\right) \left(1-\frac{\epsilon}{2}\right) \cdot \mixMMS_i(n, M \cup C)\\
&\geq (1-\epsilon) \mixMMS_i(n, M \cup C).
\end{align*}
The proof of Lemma~\ref{lem:approxMMS} is complete.\qed
\end{proof}

Lemma~\ref{lem:approxMMS} also implies that in the mixed goods setting, we can compute in polynomial time a value $\mixMMS'_i$ such that $\mixMMS_i \geq \mixMMS'_i \geq (1-\epsilon)\mixMMS_i$.

\subsubsection{Approximate $\alpha$-MMS Allocation}
% Heretofore, we haven't put any restriction on inputs so theoretically our existence result holds for any type of numbers.
% With respect to computation in practice, however, we assume that the algorithm accepts an input --- an agent's utility for each indivisible good and her value for the entire cake $C$ --- with size is no more than $L$ bits.
Now we turn to the polynomial-time algorithm for computing an approximate \amms allocation.
%In the case of irrational weights, \citet{CsehFl18} further showed that it can be decomposed into at most $n-1$ rational instances.
%As a consequence, their protocol for the rational case can be adapted to work with the irrational weights in a finite number of queries.

The algorithm is almost similar to Algorithm~\ref{alg:mixedMMS} except for
\begin{enumerate}
\item at line~\ref{ALGHomo:mms} of Algorithm~\ref{alg:mixedMMS-homoCake}, we compute the approximate values $\mixMMS'_i$, which is at most $\mixMMS_i$ and at least $(1-\epsilon) \cdot \mixMMS_i$ for each agent $i \in N$;
\item at line~\ref{ALGHomo:alpha} of Algorithm~\ref{alg:mixedMMS-homoCake}, we compute the ratio $\alpha'$ using the approximate values $\mixMMS'$, i.e., $\alpha' \leftarrow \min\left\{1, \frac12 + \min_{i \in N}\left\{\frac{u_i(C)}{2(n-1) \cdot \mixMMS'_i}\right\}\right\}$.
\end{enumerate}
A similar analysis to Lemma~\ref{lem:cakeIsEnough} shows that the new algorithm with these approximate values will still terminate.

According to Lemma~\ref{lem:approxMMS}, we know $\mixMMS_i \geq \mixMMS'_i$ for each $i \in N$, which implies that $\alpha' \geq \alpha$.
Next, for any agent $i$, by the design of the algorithm, she is guaranteed a bundle with value at least $\alpha' \cdot \mixMMS'_i \geq (1-\epsilon) \alpha' \cdot \mixMMS_i$.
Therefore the resulting allocation is $(1-\epsilon)\alpha'$-MMS.

\subsubsection{Time Complexity Analysis}
We start with the analyses of Robertson-Webb queries.
In Algorithm~\ref{alg:mixedMMS}, we need $O(n)$ value queries to obtain agents' values for the (heterogeneous) cake.
In order to compute the weights, we also need $O(n)$ value queries to get agents' values for their pieces after running \homoAlg (Algorithm~\ref{alg:mixedMMS-homoCake}).
Then, in Algorithm~\ref{alg:mixedMMS-homoCake}, in each \texttt{while}-loop, we need $O(n)$ cut queries to find the leftmost cut point (line~\ref{ALGHomo:markPoint}); overall, it takes $O(n^2)$ RW queries.

In light of Lemma~\ref{lem:approxMMS}, computing approximate MMS values takes polynomial time.
Then the only step that needs analysis is the weighted proportional allocation protocol $\weightedPROP(N, C, \mathbf{w})$ at line~\ref{ALGHetero:cakeAlloc} of Algorithm~\ref{alg:mixedMMS}.
When all weights are rational numbers, \citet{CsehFl20} gave an implementation of the protocol using $O(n \log D)$ queries, where $D$ is the common denominator of weights.
They also showed that their implementation is asymptotically the fastest possible.

We have assumed that our input has size at most $L$ bits.
Then each of the arithmetic operations in steps before line~\ref{ALGHetero:cakeAlloc} (Algorithm~\ref{alg:mixedMMS}) keeps the numbers rational with polynomial bit size.
Thus, by applying the protocol from~\citet{CsehFl20}, \weightedPROP at line~\ref{ALGHetero:cakeAlloc} of Algorithm~\ref{alg:mixedMMS} can be implemented in polynomial time.
Summarize everything together, we obtain a polynomial-time algorithm.

\begin{lemma}
Algorithm~\ref{alg:mixedMMS} runs in polynomial time with $O(n^2)$ Robertson-Webb queries and one call to the \weightedPROP oracle.
\end{lemma}

\bigskip

Sections~\ref{ssec:homoCake},~\ref{ssec:heteroCake} and~\ref{ssec:computation} together complete the proof of Theorem~\ref{thm:heteroCake}.

\subsection{Boosting the Approximation Ratio}\label{sec:improvement}
In Theorem~\ref{thm:heteroCake}, the smallest value for $\alpha$ is $\frac{1}{2}$, achieved when the resources contain only indivisible goods.
In this case, the theorem ensures that a $\frac{1}{2}$-MMS allocation always exists.
However, there is a gap between this $\frac{1}{2}$ guarantee from our result and that of the currently best-known result with only indivisible goods, which is $\gamma_I \geq \frac34 + \frac{1}{12n}$ according to~\citet{GargTa20}.
In the following, we show that a simple procedure can boost the MMS approximation ratio computed by our algorithm to (almost) match the currently best-known ratio for indivisible goods.

First, existence-wise, combining Theorem~\ref{thm:heteroCake} with Corollary~\ref{cor:gi=gm} ($\gamma_I = \gamma_M$), we can improve the ratio directly to $\max\{\alpha, \gamma_I\}$ in Theorem~\ref{thm:heteroCake}.
Next, computation-wise, suppose there exists a polynomial-time algorithm that guarantees to output a $\beta$-MMS allocation with indivisible goods for some $\beta$.
Then given a mixed good problem instance, we first compute $\alpha'$ via Theorem~\ref{thm:heteroCake} and compare it with $\beta$: if $\alpha' \geq \beta$, we directly apply Theorem~\ref{thm:heteroCake}; otherwise, we cut the cake $C$ into small intervals, each valued at most $\frac{\epsilon \cdot u_i(C)}{2n}$ for each agent $i$, and use the $\beta$-MMS algorithm to obtain the allocation of this instance with only indivisible goods.
In summary, we have the following strengthened result:

\begin{theorem}\label{thm:heteroCakeBoost}
A $\max \{\alpha, \gamma_I\}$-MMS allocation with mixed goods always exists for any number of agents.

In addition, if there exists a polynomial-time algorithm that can always output a $\beta$-MMS allocation with indivisible goods, then for any constant $\epsilon > 0$, there is another polynomial-time algorithm that computes a $(1-\epsilon) \max\{\alpha', \beta\}$-MMS allocation with mixed goods.
\end{theorem}

The proof of Theorem~\ref{thm:heteroCakeBoost} utilizes the proof of Lemma~\ref{lem:approxMMS} and is straightforward to prove.
%The proof of Theorem~\ref{thm:heteroCakeBoost} is simple and utilizes the proof of Lemma~\ref{lem:approxMMS}, which is omitted here.
The currently best lower bound of $\gamma_I$ is $\frac34 + \frac{1}{12n}$ and the currently best-known value of $\beta$ is $\frac34$, both are due to \citet{GargTa20}.
Any better lower bound of $\gamma_I$ and value of $\beta$ found in the future would immediately imply a better MMS approximation guarantee in the mixed goods setting as well.

% \reviews{Review: The proposed $\alpha$-MMS algorithm is only interesting for instances in which alpha can be higher than $3/4$ as otherwise a rather simple adaptation of the algorithm by Ghoshi et al. can always output an almost $3/4$-MMS allocation. I think you should start your discussion from the algorithm of Ghoshi et al. and argue that it can always give a $3/4$-MMS allocation. Then, discuss that you are actually able to outperform it in some case, and hence have an algorithm that gets the $\max{\alpha,\beta}$. \\
% Shengxin: I think it is better to keep the current flow unchanged.}

\section{Relation of MMS and EFM}\label{sec:relation}
Proportionality fairness, and its generalization, MMS, are often compared to another well-studied fairness notion of \emph{envy-freeness (EF)}.
It is known that with only divisible goods, envy-freeness implies proportionality but not vice versa.
%With only indivisible goods, the relaxed notion of EF, known as \emph{envy-freeness up to one item (EF1)}, and the relaxed notion of proportionality, MMS, do not imply each other~\citep{CaragiannisKuMo19}.
With only indivisible goods, envy-freeness is often relaxed to \emph{envy-freeness up to one item (EF1)} and MMS is often considered as a relaxation of proportionality.
It is known from~\citep{CaragiannisKuMo19} that neither EF1 nor MMS implies the other.
In a recent work, \citet{BeiLiLi21} proposed a new envy-freeness notion, termed \emph{envy-freeness for mixed goods (EFM)}, that generalizes both EF and EF1 to the mixed goods setting.
%We include the definition of EFM as follows.

We include the definition of EF, EF1, and EFM here for the sake of being self-contained.

\begin{definition}[EF]
An allocation $\mathcal{A}$ is said to satisfy \emph{envy-freeness (EF)} if for any pair of agents $i, j \in N$, $u_i(A_i) \geq u_i(A_j)$.
\end{definition}

\begin{definition}[EF1]
With indivisible goods, an allocation $\mathcal{A}$ is said to satisfy \emph{envy-freeness up to one good (EF1)} if $\forall i, j \in N$, $\exists g \in A_j$, such that $u_i(A_i) \geq u_i(A_j \setminus \{g\})$.
\end{definition}

\begin{definition}[EFM]
An allocation $\mathcal{A}$ is said to satisfy \emph{envy-freeness for mixed goods (EFM)} in the sense that for any $i, j \in N$,
\begin{itemize}
\item if $j$'s bundle consists of only indivisible goods, there exists $g \in A_j$ such that $u_i(A_i) \geq u_i(A_j \setminus \{g\})$;
\item otherwise, $u_i(A_i) \geq u_i(A_j)$.
\end{itemize}
\end{definition}

% We first investigate the relation between EFM and (full) MMS.
% In particular, we show that:

% \begin{lemma}\label{lem:MMS_EFM}
% Neither MMS nor EFM implies the other.
% \end{lemma}

% \begin{proof}
% First note that an MMS allocation does not always exist with mixed goods, but EFM does \citep{BeiLiLi20}.
% This immediately shows that EFM does not imply MMS.
% %
% On the other hand, as EFM is EF1 for indivisible gods, and MMS does not imply EF1~\citep{CaragiannisKuMo16} in this case.
% % On the other hand, we show via the following example that MMS also does not imply EFM.
% % Consider an instance with 2 agents $N = \{1, 2\}$, the set of indivisible goods $M = \{a, b\}$, and a homogeneous cake $C$.
% % We list below the utilities of each agent for the goods:
% % \begin{description}
% % \item[Agent 1:] $u_1(a) = 2, u_1(b) = 1, u_1(C) = 0.5$.
% % \item[Agent 2:] $u_2(a) = 1.5, u_2(b) = 2, u_2(C) = 0$.
% % \end{description}
% % One can check that the allocation $(\{b, C\}, \{a\})$ satisfies MMS.
% % However, this allocation is not EFM because agent 2 envies agent 1, whose bundle contains divisible good.
% \end{proof}

%We investigate the relation between MMS and EFM in the mixed goods setting.
As, with only indivisible goods, EFM reduces to EF1, it is obvious to see that neither EFM nor MMS implies the other.
We then consider the relation between EFM and the approximation of MMS, focusing on what approximation ratio of MMS can be achieved by an EFM allocation.

On the one hand, when all goods are divisible, EFM (or EF) is always $1$-MMS (or proportionality).
On the other hand, when all goods are indivisible, \citet{AmanatidisBiMa18} showed that any EFM (or EF1) allocation is always $\frac{1}{n}$-MMS and this approximation ratio is tight.
Then, with mixed goods, one might ask if an EFM allocation would have the MMS approximation ratio laying between $\frac{1}{n}$ and $1$.
Our next lemma confirms this conjecture.

\begin{lemma}\label{lem:EFMvsMMS}
Given any mixed goods instance $\langle N, M \cup C \rangle$, for any EFM allocation $(A_1, A_2, \dots, A_n)$ and any agent $i \in N$, we have $$v_i(A_i) \geq \frac{\mixMMS_i(n, M) + v_i(C)}{n} \geq \frac{\mixMMS_i(n, M \cup C)}{n}.$$
\end{lemma}

The proof is a direct generalization of the proof of Proposition 3.6 in~\citep{AmanatidisBiMa18}.

From Lemma~\ref{lem:EFMvsMMS}, we know that EFM implies \amms where $\alpha$ is a monotonically increasing function that depends on the agent's value on the whole cake.
In other words, one can directly utilize the EFM allocation to obtain an \amms allocation with $\alpha$ varied from $1/n$ (when goods are indivisible only) to $1$ (when goods are divisible only).
On the other hand, our result in Section~\ref{sec:enoughCakeHelp} shows that we can always have an \amms allocation with $\alpha$ ranging from $1/2$ to $1$.

%% MMS does not imply EFM: three-agent case.
% \begin{comment}
% \begin{example}[MMS does not imply EFM]
% Let the set of agen{}ts be $N = \{1, 2, 3\}$, the set of goods be $M = \{a\}$ and the cake be $C = [0, 1]$ which is homogeneous to every agent.
% The valuations $v_i$ of each agent $i \in N$ are listed as follows:
% \begin{description}
% \item[Agent 1:] $u_1(a) = 6, u_1(C) = 30$.{}
% \item[Agent 2:] $u_2(a) = 3, u_2(C) = 30$.
% \item[Agent 3:] $u_3(a) = 3, u_3(C) = 30$.
% \end{description}
% The allocation $(\{a, [0, 7/30]\}, \{[7/30, 19/30]\}, \{19/30, 1\})$ satisfies MMS, but is not EFM because agent 3 envies agent 2.
% \end{example}
% \end{comment}

\section{Conclusion and Future Work}\label{sec:conclusion}
In this paper, we study the extent to which we can find approximate MMS allocations when the resources contain both divisible and indivisible goods.
We analyze the relation of the worst-case MMS approximation guarantees between mixed goods instances and indivisible goods instances.
We also present an algorithm to produce an \amms allocation for any number of agents, where $\alpha$ monotonically increases in terms of the ratio between agents' values for the divisible goods and their MMS values.
% We also show a polynomial time approximation algorithm to compute a $(1-\epsilon)\alpha$-MMS allocation for any constant $\epsilon > 0$.

For future work, it would be interesting to improve the MMS approximation guarantee with mixed goods.
It would also be interesting to investigate other fairness notions in the mixed goods setting --- notions we mentioned like pairwise MMS and groupwise MMS fit the setting well.
%Also, as we mentioned, an interesting open question would be studying what approximation of MMS an EFM allocation implies.
Another direction is to study fair allocations in the mixed goods setting in conjunction with economic efficiency notions such as Pareto optimality.
There have been some preliminary results, for instance, in the context of determining the compatibility of EFM and PO by \citet{BeiLiLi21}.
One could also further generalize the mixed resources model to capture other practical scenarios.
One natural consideration is to let agents have either positive or negative utility for each item.
This has been done, for example, when items are all divisible or indivisible~\citep{BogomolnaiaMoSa17,AzizCaIg19,ChaudhuryGaMc21}.

%% Bibliography
\bibliographystyle{plainnat}
\bibliography{JAAMAS}
\end{document}